\documentclass[11pt]{article}
\usepackage[margin=1in]{geometry}
\usepackage{graphicx}
\usepackage[pdftex,colorlinks=true,linkcolor=blue,citecolor=blue,urlcolor=black]{hyperref}
\usepackage{amsmath, amsthm, amssymb}
\usepackage{subfigure}
\usepackage{comment}
\usepackage{url}
\usepackage{pdflscape}
\usepackage[ruled,lined,linesnumbered]{algorithm2e}
\usepackage{tikz}

\newcommand{\nc}{\newcommand}
\nc\cD{\mathcal{D}}
\nc\eps{\epsilon}
\nc{\ra}{\rightarrow}
\nc{\eq}[1]{(\ref{eq:#1})}
\nc{\eqs}[2]{(\ref{eq:#1}) and (\ref{eq:#2})}
\nc{\ot}{\otimes}

\newcommand{\C}{\mathbb{C}}
\newcommand{\Z}{\mathbb{Z}}

\newcommand{\E}{\mathbb{E}}
\newcommand{\F}{\mathbb{F}}

\newcommand{\ket}[1]{| #1 \rangle}
\newcommand{\bra}[1]{\langle #1|}
\newcommand{\ip}[2]{\langle #1|#2 \rangle}
\newcommand{\hsip}[2]{\langle #1,#2 \rangle}
\newcommand{\proj}[1]{| #1 \rangle \langle #1 |}

\newcommand{\bracket}[3]{\langle #1|#2|#3 \rangle}

\renewcommand{\Re}{\operatorname{Re}}

\DeclareMathOperator{\tr}{tr}

\newcommand{\be}{\begin{equation}}
\newcommand{\ee}{\end{equation}}
\newcommand{\bea}{\begin{eqnarray}}
\newcommand{\eea}{\end{eqnarray}}
\newcommand{\bes}{\begin{equation*}}
\newcommand{\ees}{\end{equation*}}
\newcommand{\beas}{\begin{eqnarray*}}
\newcommand{\eeas}{\end{eqnarray*}}

\makeatletter
\newtheorem*{rep@theorem}{\rep@title}
\newcommand{\newreptheorem}[2]{%
\newenvironment{rep#1}[1]{%
 \def\rep@title{#2 \ref{##1} (restated)}%
 \begin{rep@theorem}}%
 {\end{rep@theorem}}}
\makeatother

\newtheorem{thm}{Theorem}
\newtheorem*{thm*}{Theorem}
\newtheorem{cor}[thm]{Corollary}

\newtheorem{lem}[thm]{Lemma}
\newtheorem*{lem*}{Lemma}
\newtheorem{prop}[thm]{Proposition}

\newreptheorem{thm}{Theorem}
\newreptheorem{lem}{Lemma}
\newreptheorem{cor}{Corollary}

\newcommand{\boxalgm}[3]{
\renewcommand{\figurename}{Algorithm}
\begin{figure}[t]
\noindent \framebox{
\begin{minipage}{0.96\textwidth}
#3
\end{minipage}
}
\caption{#2}
\label{#1}
\end{figure}
\renewcommand{\figurename}{Figure}
}

\newcommand{\thmref}[1]{Theorem~\ref{thm:#1}}

\begin{document}

\title{Sequential measurements, disturbance and property testing}
\author{Aram W. Harrow\thanks{Center for Theoretical Physics, Massachusetts Institute of Technology, USA; {\tt aram@mit.edu}}, \ Cedric Yen-Yu Lin\thanks{Joint Center for Quantum Information and Computer Science, University of Maryland, College Park, MD, USA; {\tt cedricl@umiacs.umd.edu}} \ and Ashley Montanaro\thanks{School of Mathematics, University of Bristol, UK; {\tt ashley.montanaro@bristol.ac.uk}.}}
\maketitle

\begin{abstract}
We describe two procedures which, given access to one copy of a quantum state and a sequence of two-outcome measurements, can distinguish between the case that at least one of the measurements accepts the state with high probability, and the case that all of the measurements have low probability of acceptance. The measurements cannot simply be tried in sequence, because early measurements may disturb the state being tested. One procedure is based on a variant of Marriott-Watrous amplification. The other procedure is based on the use of a test for this disturbance, which is applied with low probability.
We find a number of applications:
\begin{itemize}
\item Quantum query complexity separations in the property testing model for testing isomorphism of functions under group actions. We give quantum algorithms for testing isomorphism, linear isomorphism and affine isomorphism of boolean functions which use exponentially fewer queries than is possible classically, and a quantum algorithm for testing graph isomorphism which uses polynomially fewer queries than the best algorithm known.

\item Testing properties of quantum states and operations. We show that any finite property of quantum states can be tested using a number of copies of the state which is logarithmic in the size of the property, and give a test for genuine multipartite entanglement of states of $n$ qubits that uses $O(n)$ copies of the state. We also show that equivalence of two unitary operations under conjugation by a unitary picked from a fixed set can be tested efficiently. This is a natural quantum generalisation of testing isomorphism of boolean functions.

\item Correcting an error in a result of Aaronson on de-Merlinizing quantum protocols. This result claimed that, in any one-way quantum communication protocol where two parties are assisted by an all-powerful but untrusted third party, the third party can be removed with only a modest increase in the communication cost. We give a corrected proof of a key technical lemma required for Aaronson's result.
\end{itemize}
\end{abstract}


\section{Introduction}

In quantum mechanics, measuring the state of a system generally disturbs it. This effect can be harnessed for practical purposes -- such as in the field of quantum cryptography~\cite{gisin02} -- but is often an annoyance. In particular, the following question is nontrivial: Given one copy of a state $\ket{\psi}$ and a set of $n$ measurements, determine whether at least one of the measurements would accept $\ket{\psi}$ with high probability. If one simply tries the measurements one after the other, early measurements may disturb $\ket{\psi}$, leading later measurements to give an incorrect outcome.

A vivid example of this phenomenon is the quantum anti-Zeno effect~\cite{aharonov80,schieve89,kaulakys97}. Just as the more familiar quantum Zeno effect~\cite{misra77} describes a situation where the time-evolution of a quantum system is halted by frequent measurements, the anti-Zeno effect describes a situation where time-evolution is {\em caused} by measurement\footnote{One could also think of this effect as the standard Zeno effect in a rotating reference frame.}. For example, consider the sequence $M_1,\dots, M_n$ of 2-outcome measurements of one qubit such that the $k$'th measurement is specified by the pair $\{ I -\proj{\psi_k}, \proj{\psi_k} \}$, where we associate the first outcome with acceptance, the second with rejection, and
\[ \ket{\psi_k} = \cos\left(\frac{\pi k}{2n}\right) \ket{0} + \sin\left(\frac{\pi k}{2n}\right) \ket{1}. \]
Then, if we have the state $\ket{\psi_k}$ and apply the measurement $M_{k+1}$, the probability of rejection is
\beas
|\ip{\psi_k}{\psi_{k+1}}|^2 &=& \bigg(\cos\left( \frac{\pi k }{2n}\right) \cos\left(\frac{\pi (k+1)}{2n}\right) 
+ \sin\left(\frac{ \pi k }{2n}\right) \sin\left(\frac{\pi (k+1)}{2n}\right)\bigg)^2\\
&=& \cos\left(\frac{\pi}{2n}\right)^2\\
&=& 1 - O(1/n^2)
\eeas
and the residual state following rejection is $\ket{\psi_{k+1}}$. Therefore, if each of the $n$ measurements is performed in order, starting with the initial state $\ket{0} = \ket{\psi_0}$, the probability of ever seeing an ``accept'' measurement outcome is $O(1/n)$. However, the final state is $\ket{\psi_n} = \ket{1}$, which is orthogonal to the initial state; and indeed if the final measurement $M_n$ were performed on $\ket{0}$, it would accept with certainty.

Here we describe quantum procedures that combat this effect:

\begin{thm}\label{thm:informal}
Let $\Lambda_1,\dots,\Lambda_n$ be a sequence of projective measurement operators, where $\Lambda_i$ corresponds to the two-outcome measurement $M_i = \{ \Lambda_i, I - \Lambda_i \}$. Let $\rho$ be a state such that either there exists $i$ such that $\tr \Lambda_i \rho = \Omega(1)$, or $\E_j[\tr \Lambda_j \rho] = o(1/n)$. Then there is a test that uses one copy of $\rho$ and: in the first case, accepts with probability $\Omega(1)$; in the second case, accepts with probability $o(1)$.
\end{thm}

\thmref{informal} is restated more precisely below as Corollary \ref{cor:seqmeasproptest}.  A variant of this result was called a ``quantum OR bound'' by Aaronson~\cite{aaronson06}, but the proof there was incorrect because it neglected the effects of disturbance in the ``accept'' case.  The term ``OR bound'' refers to the fact that we would like to design a measurement that accepts if any one of the $M_i$ are likely to accept.

We give two procedures for \thmref{informal}. One procedure performs a crude form of eigenvalue estimation on the POVM operator $\frac{1}{n} \sum_j \Lambda_j$. In the first case $\rho$ has constant overlap with the +1-eigenspace of $\Lambda_i$, and therefore has constant overlap with the space of eigenvectors of $\frac{1}{n} \sum_j \Lambda_j$ with eigenvalue at least $\Omega(1/n)$. We employ a variant of the Marriott-Watrous gap amplification procedure~\cite{marriott05} to determine when this is the case. The second procedure is based on repeatedly selecting a random measurement from the sequence to perform, but also incorporating a ``disturbance test'', which is performed with low probability. This allows us to detect whether the residual state of the system after some measurements is far from its initial state. 

These procedures turn out to have a number of applications.  The common theme behind these applications is that we are given a large number of tests and promised that the maximum acceptance probability of these tests is either large or small.   For example, we might be given a class of symmetries and a state that is (a) either invariant under some symmetry or (b) far from invariant under all such symmetries.  In case (b) each symmetry test is unlikely to accept and therefore causes little damage.  However, if case (a) holds then besides the symmetry that the state exactly satisfies there may be many others where the symmetry test has any probability of accepting between 0 and 1.  Thus naive sequential measurement could substantially damage the state.  Our \thmref{informal} is designed precisely to address this possibility.  We now describe the applications in detail.

\subsection{Property testing}

The field of property testing aims to find super-efficient algorithms for determining whether a function has a certain property, given the promise that it either has the property, or is far from having the property. The goal is to achieve this with the minimal number of queries to the function. A number of properties are known where quantum testers outperform classical testers -- sometimes exponentially or super-exponentially. See~\cite{montanaro13c} for a review. 

Here we present query-efficient quantum testers for properties which can be expressed as isomorphism of functions under the action of a permutation group. Let $G$ be a nontrivial permutation group acting on a finite set $X$, and consider two functions $f,g:X \rightarrow Y$ for some finite set $Y$. We say that $f$ and $g$ are isomorphic if there exists $\sigma \in G$ such that $g(x) = f(\sigma(x))$ for all $x \in X$, and for conciseness sometimes write $g = f\circ\sigma$. On the other hand, we say that $f$ and $g$ are $\epsilon$-far from isomorphic if, for all $\sigma \in G$, $|\{x \in X: g(x) \neq f(\sigma(x))\}| \ge \epsilon |X|$. We say that an algorithm is an $\epsilon$-tester for the property of $G$-isomorphism if it distinguishes between these two cases with success probability at least $2/3$. This success probability can be improved to $1-\delta$, for any $\delta > 0$, by running the algorithm $O(\log 1/\delta)$ times and taking the majority vote. Below we discuss concrete examples of this abstract problem.

The problem of testing $G$-isomorphism was studied by Babai and Chakraborty~\cite{babai10}, who considered both the case where $g$ is known in advance (which they call the ``query-1'' case), and the case where both $f$ and $g$ are unknown (the ``query-2'' case). Here we only consider the case where $f$ and $g$ are both unknown. For this case, Babai and Chakraborty gave an almost tight classical bound of $\widetilde{\Theta}(\sqrt{|X| \log |G|})$ queries for testers with 1-sided error in the case where the action of $G$ is primitive, i.e.\ does not preserve a nontrivial partition of $X$, and a classical upper bound of $O(\sqrt{|X| \log |G|})$ queries which holds for all permutation groups $G$ and also has only 1-sided error. The best general lower bound they give for classical testers with 2-sided error is substantially weaker: $\Omega(\epsilon \log n)$ queries for any transitive group $G$ of order $2^{O(n^{1-\epsilon})}$. However, substantially stronger bounds can be shown for specific group actions, as discussed below.

Here we prove the following result:

\begin{thm}
\label{thm:giso}
For any set of permutations $G$, there is a quantum $\epsilon$-tester for $G$-isomorphism which makes $O((\log |G|) / \epsilon)$ queries.
\end{thm}

Observe that there is no dependence on $|X|$, unlike the best known classical upper bounds; and that going slightly beyond the usual definition of $G$-isomorphism~\cite{babai10}, Theorem \ref{thm:giso} holds when $G$ is an arbitrary subset of a permutation group, rather than needing to be itself a group. Theorem \ref{thm:giso} encompasses a number of special cases:

\begin{itemize}
\item {\bf Isomorphism of boolean functions.} Two boolean functions $f,g:\{0,1\}^n \rightarrow \{0,1\}$ are said to be isomorphic if there exists $\sigma \in S_n$ such that $g(x) = f(\sigma(x))$ for all $x \in \{0,1\}^n$, where $\sigma(x)_i = x_{\sigma(i)}$. It is known that testing isomorphism of two unknown boolean functions up to constant accuracy requires $\Omega(2^{n/2} / n^{1/4})$ classical queries, with a nearly-matching upper bound of $O(2^{n/2} \sqrt{n \log n})$ queries~\cite{alon13}. Theorem \ref{thm:giso} implies (taking $G = S_n$, $X = \{0,1\}^n$) that the quantum query complexity of $\epsilon$-testing isomorphism of boolean functions is $O((n \log n) / \epsilon)$, which is exponentially smaller.

\item {\bf Graph isomorphism.} Two graphs $G$ and $H$ on $n$ vertices are said to be isomorphic if there exists a bijection $\sigma$ between the vertices such that vertices $u$ and $v$ are adjacent in $G$ if and only if $\sigma(u)$ and $\sigma(v)$ are adjacent in $H$. The best classical query complexity known for testing graph isomorphism is $\widetilde{O}(n^{5/4})$~\cite{fischer08}, though the current best classical lower bound is only $\Omega(n)$~\cite{fischer08}. A quantum algorithm for testing graph isomorphism which makes $\widetilde{O}(n^{7/6})$ queries was given by Chakraborty et al.~\cite{chakraborty13a}, as well as a quantum lower bound of $\Omega(n^{1/3})$ queries. Theorem \ref{thm:giso} implies (taking $G = S_n$,  $X=[n]\times [n]$ and $Y=\{0,1\}$) 
 that the quantum query complexity of $\epsilon$-testing graph isomorphism is $O((n \log n) / \epsilon)$.

\item {\bf Linear and affine isomorphism of boolean functions.} Two boolean functions $f,g:\{0,1\}^n \rightarrow \{0,1\}$ are said to be linear-isomorphic if there exists a non-singular linear transformation $A \in GL_n(\F_2)$ such that $g(x) = f(Ax)$ for all $x \in \{0,1\}^n$. Testing linear isomorphism was studied by Grigorescu, Wimmer and Xie~\cite{grigorescu13}, who gave an $\Omega(n^2)$ lower bound for testing linear isormorphism to a fixed function. It is not difficult to prove a substantially stronger $\Omega(2^{n/2})$ classical lower bound for testing linear isomorphism between two unknown functions; the argument is closely related to the lower bound for Simon's problem (qv) and we include it in Appendix~\ref{app:linisobound}. 
Theorem \ref{thm:giso} implies (taking $G = GL_n(\F_2)$, $X = \{0,1\}^n$) that the quantum query complexity of $\epsilon$-testing linear isomorphism of boolean functions is $O(n^2 / \epsilon)$, which is exponentially smaller.

A similar $O(n^2 / \epsilon)$ quantum upper bound holds for testing affine isomorphism, i.e.\ the existence of a pair $A \in GL_n(\F_2)$, $b \in \F_2^n$ such that $g(x) = f(Ax + b)$ for all $x$. Here an $\Omega(2^{n/2})$ classical lower bound is immediate from the known $\Omega(2^{n/2})$ lower bound on the classical query complexity of the property-testing variant of Simon's problem~\cite{buhrman08}.

\item {\bf Hidden subgroup problems.} Imagine we have access to a function $f:G \rightarrow Y$ for some group $G$ and some set $Y$, and are promised that either $f$ is constant on cosets of a nontrivial subgroup $H \le G$, or $f$ is far from any such function. This is a property-testing version of the well-studied hidden subgroup problem~\cite{lomont04} for the group $G$. Take $X$ to be $G$, on which $G$ acts in the natural way. Then the pair of functions $(f,f)$ are $(G \setminus \{e\})$-isomorphic if and only if $f$ is constant on cosets of a nontrivial subgroup $H \le G$. Theorem \ref{thm:giso} implies that the query complexity of this problem is $O((\log |G|) / \epsilon)$. A query-efficient tester for this property with the same complexity parameters was previously described by Friedl et al.~\cite{friedl09} in the case where $H$ is promised to be a {\em normal} subgroup of $G$; in addition, their tester is time-efficient when $G$ is abelian. Note that the promise here is a bit different to the standard HSP. There, we are promised that $f$ is constant on cosets of a subgroup $H \le G$, and distinct on each coset. A quantum tester for this variant where $G = \Z_N$, corresponding to a periodicity determination problem, was given in~\cite{chakraborty13a}. In this case, the query complexity can be reduced to $O(1)$ if the additional constraint is imposed that the period of $f$ is $O(\sqrt{N})$.

It is known that the quantum query complexity of the property-testing version of Simon's problem, a special case of the hidden subgroup problem for $G = \Z_2^n$, obeys an $\Omega(\log |G|)$ lower bound~\cite{koiran05} (see also the discussion in~\cite{montanaro13c}). This implies that the dependence on $|G|$ of Theorem \ref{thm:giso} is optimal.
\end{itemize}

Following the completion of this work, Alexander Belov (personal communication) has notified us that Theorem \ref{thm:giso} can also be proven by a direct construction of a solution to the quantum adversary bound semidefinite program. 


\subsection{Testing quantum properties}

We can also apply the ideas underlying our tester to testing properties of quantum states and operations. The notion of testing properties of classical data can naturally be extended to pure quantum states (unit vectors in $\C^d$) as follows. Let $\mathcal{P}$ be a subset of $\C^d$. We are given the ability to produce copies of an initially unknown state $\ket{\psi}$, and asked to distinguish between the following two cases with success probability $2/3$: either $\ket{\psi} \in \mathcal{P}$, or $\inf_{\ket{\phi} \in \mathcal{P}} \| \psi - \phi \|_{\operatorname{tr}} \ge \epsilon$. Here $\| \cdot \|_{\operatorname{tr}}$ is the trace distance,
\be \label{eq:td} \| \proj{\psi} - \proj{\phi}  \|_{\operatorname{tr}} = \frac{1}{2} \tr |\proj{\psi} - \proj{\phi}| = \sqrt{1 - |\ip{\psi}{\phi}|^2}, \ee
and we use the notation $\psi = \proj{\psi}$. We say that an algorithm which achieves these bounds is a quantum $\eps$-tester for $\mathcal{P}$. Some interesting properties of quantum states (such as productness and permutation-invariance) can be tested in this framework~\cite{montanaro13c}.

Wang~\cite{wang11} has given a tester which tests membership of a state $\ket{\psi}$ in an {\em arbitrary} finite subset $\mathcal{P}$ of quantum states using $O(\log |\mathcal{P}|)$ copies of $\ket{\psi}$, assuming that all the states in $\mathcal{P}$ are far apart from each other. Formally, Wang's result is as follows:

\begin{thm}[Wang~\cite{wang11}]
\label{thm:wang}
Let $\mathcal{P}$ be a finite subset of the unit sphere in $\C^d$ such that $\min_{\ket{\phi},\ket{\phi'} \in \mathcal{P}} \| \phi - \phi'  \|_{\operatorname{tr}} = \zeta$. Then, for any $\eps > 0$, there is a test which accepts every state in $\mathcal{P}$ with certainty, rejects every state $\ket{\psi}$ such that $\min_{\ket{\phi} \in \mathcal{P}} \| \psi - \phi  \|_{\operatorname{tr}}\ge \epsilon$ with probability at least $2/3$, and uses $O((\log |\mathcal{P}|) \max\{ \epsilon^{-2}, \zeta^{-2} \})$ copies of the input state.
\end{thm}

Wang applied this result to testing finite subsets of the unitary group, and in particular to testing permutations~\cite{wang11}.  Here we can improve Theorem \ref{thm:wang} by removing the dependence on the minimum distance $\zeta$, at the expense of introducing two-sided error:

\begin{thm}
\label{thm:teststates}
Let $\mathcal{P}$ be a finite subset of the unit sphere in $\C^d$. Then, for any $\epsilon > 0$, there is a quantum $\eps$-tester for $\mathcal{P}$ which uses $O((\log |\mathcal{P}|) /\epsilon^2)$ copies of the input state.
\end{thm}

This resolves Question 6 in~\cite{montanaro13c}. We remark that a similar bound to Wang's (albeit with two-sided error) can be obtained simply by reducing the problem of testing membership of $\ket{\psi}$ in $\mathcal{P}$ to {\em identifying} $\ket{\psi}$, given the promise that it is contained in $\mathcal{P}$~\cite{harrow12a}. Theorem \ref{thm:teststates} goes beyond this idea, and efficiently tests membership in sets of states $\mathcal{P}$ for which identifying a member of $\mathcal{P}$ is more challenging, because there exist states in $\mathcal{P}$ which are close to each other.

The ability to test properties of states naturally extends to testing properties of operations, i.e.\ unitary operators or quantum circuits~\cite{wang11,montanaro13c}. Indeed, using a standard correspondence between states and operations known as the Choi-Jamio\l kowski isomorphism, any tester for a property of quantum states gives a tester for a corresponding property of unitaries. Now the relevant distance measure can be defined in terms of the Hilbert-Schmidt inner product $\langle A, B \rangle = \tr A^\dag B / d$ for $d$-dimensional operators $A$, $B$, as $D(A,B) = \sqrt{1-|\hsip{A}{B}|^2}$~\cite{montanaro13c} (compare (\ref{eq:td})). This is an average-case measure of distance, as for unitary operators $U$, $V$,
\[ D(U,V)^2 \approx \int \|U\proj{\psi}U^\dag - V\proj{\psi}V^\dag\|^2_{\operatorname{tr}}\,d\psi, \]
where the integral is taken over the Haar (uniform) measure on the set of pure quantum states $\ket{\psi}$~\cite{montanaro13c}; it is also closely related to the distance between $A$ and $B$ in the Schatten 2-norm~\cite{montanaro13c}. Let $\mathcal{P}$ be a subset of $U(d)$, the set of $d$-dimensional unitary operators. We are given access to an initially unknown operator $U \in U(d)$, and asked to distinguish between the following two cases with success probability $2/3$: either $U \in \mathcal{P}$, or $\inf_{V \in \mathcal{P}} D(U,V) \ge \epsilon$. As in the case of properties of quantum states, we say that an algorithm which achieves these bounds is a quantum $\eps$-tester for $\mathcal{P}$. 

Using Theorem \ref{thm:teststates}, we obtain the following corollary:

\begin{cor}
\label{cor:testunitaries}
Let $\mathcal{P}$ be a finite subset of $U(d)$. Then, for any $\epsilon > 0$, there is a quantum $\eps$-tester for $\mathcal{P}$ which makes $O((\log |\mathcal{P}|) /\epsilon^2)$ uses of the input unitary operator.
\end{cor}

Another consequence of these ideas is an efficient algorithm for a generalisation of the notion of $G$-isomorphism of functions to isomorphism of unitary operators\footnote{We would like to thank Noah Linden for suggesting this application.}. Imagine we have a known set $S = U_1,\dots,U_n$ of unitary operators, and two unknown unitary operators $V$ and $W$. We would like to determine whether there exists $U \in S$ such that $U V U^\dag = W$. More precisely, our task is to distinguish between two cases: either there exists $U \in S$ such that $U V U^\dag = W$, or for all $U \in S$, $D(U V U^\dag, W) \ge \epsilon$. We call this property unitary $S$-isomorphism. For example, $V$ and $W$ could be quantum circuits on $n$ qubits, $S$ could be the set of all permutations of $n$ qubits, and we might like to determine whether there exists $\sigma \in S$ such that $\sigma V \sigma^{-1} = W$. This is one natural quantum generalisation of the property of isomorphism of boolean functions.

\begin{thm}
\label{thm:uiso}
For any $S$, there is a quantum $\epsilon$-tester for unitary $S$-isomorphism which makes $O((\log |S|) / \epsilon^2)$ uses of the input unitaries $V$ and $W$.
\end{thm}


\subsection{Testing genuine multipartite entanglement}
One important property of quantum states is that of being entangled.   
A pure state $\ket{\psi}$ on $n$ qudits is said to be product (resp.~entangled) across the cut $S:S^c$ (for $S$ a nontrivial subset of $[n]$) if $\ket{\psi} = \ket{\alpha}_S \ot \ket{\beta}_{S^c}$ for some states $\ket{\alpha},\ket{\beta}$ (resp.~$\ket{\psi}$ cannot be written in this way).   If such an $S$ exists then we say that $\ket\psi$ is product across some cut.   If no such $S$ exists then we say that $\ket{\psi}$ has ``genuine multipartite entanglement'' or ``genuine $n$-partite entanglement''; this term reflects the fact that an $n$-partite state might be describable entirely in terms of entangled states on smaller subsets of systems.

In \cite{harrow10,harrow13} two of us analyzed a property tester for a related question, of whether $\ket{\psi}$ is equal to an $n$-partite product state $\ket{\alpha_1}\otimes \cdots \otimes \ket{\alpha_n}$ or far from any such state.  This property tester -- known as the ``product test'' -- can also be used to test whether $\ket{\psi}$ is entangled across any fixed cut $S:S^c$.  With \thmref{informal} we can extend this to test whether entanglement exists across all cuts simultaneously.

\begin{thm}\label{thm:genuine}
There is a quantum $\eps$-tester for the property of an $n$-partite state being product across some cut.  The tester uses $O(n/\eps^2)$ copies of the state.
\end{thm}

We note that this result does not follow from \thmref{teststates} because the set of product states is infinite, and even defining an $\eps$-net over the set of product states would cause the sample complexity to depend on the dimensions of the subsystems.


\subsection{De-Merlinizing quantum protocols}

Our final application, following work of Aaronson~\cite{aaronson06}, is to a problem in quantum communication complexity~\cite{buhrman10}.  Let $f:\{0,1\}^n \times \{0,1\}^m \rightarrow \{0,1\}$ be a (possibly partial) boolean function. We imagine that the input to $f$ is divided into two parts, the first of which is given to Alice, the second to Bob. Their goal is to output $f(x,y)$ for given $x$ and $y$ with the minimum amount of communication. Let $Q^1(f)$ denote the bounded-error one-way quantum communication complexity of $f$: the minimal number of qubits Alice needs to send to Bob in order to achieve worst-case success probability of $2/3$.

A more general model~\cite{aaronson06} is to allow Alice and Bob to be assisted by Merlin, who provides a quantum witness $\ket{\phi}$ to Bob, and aims to convince him that $f(x,y) = 1$. Let $\ket{\psi_x}$ be the state which Alice sends to Bob on input $x$. We demand that:
\begin{itemize}
\item If $f(x,y)=1$, then there exists $\ket{\phi}$ such that Bob accepts the triple $(y, \ket{\psi_x}, \ket{\phi})$ with probability at least $2/3$;
\item If $f(x,y)=0$, then for all $\ket{\phi}$, Bob accepts $(y, \ket{\psi_x}, \ket{\phi})$ with probability at most $1/3$.
\end{itemize}
Let QMA$^1_w(f)$ denote the minimal $a$ such that there exists a protocol of this form where Alice sends $a$ qubits to Bob, and Merlin sends $w$ qubits to Bob. Then it was claimed in~\cite{aaronson06} that:

\begin{thm}
\label{thm:qma1}
For all partial or total boolean functions $f$, and all $w \ge 2$,
\[ Q^1(f) = O(\operatorname{QMA}^1_w(f) \cdot w \log^2 w). \]
\end{thm}

Several applications of this result were also given in~\cite{aaronson06}, to random access coding, one-way communication and complexity theory. The basic idea behind the proof of Theorem \ref{thm:qma1} is that first Alice and Bob's protocol can be amplified such that the soundness error (the probability of incorrectly accepting) can be made exponentially small in $w$, while keeping the completeness error (the probability of incorrectly rejecting) at most $1/3$. Then Bob can replace his use of the witness $\ket{\phi}$ from Merlin with simply looping over all possible computational basis states of Merlin's witness register and trying each such state as a witness in turn. In a ``yes'' case ($f(x,y)=1$), at least one of these will have a sufficiently high probability of acceptance that this can be distinguished from a ``no'' case. This procedure requires $O(2^w)$ measurements to be applied to $\ket{\psi_x}$ (actually slightly more because of the amplification step).

Unfortunately, the proof of a key lemma required for Theorem \ref{thm:qma1} in~\cite{aaronson06} (Lemma 14, the ``Quantum OR bound'') does not appear to be correct. One step of the proof claims the following: given a sequence of $t$ 2-outcome measurements performed on an initial state $\rho$ to leave a final state $\rho_t$, if $\| \rho_t - \rho\|_{\operatorname{tr}} > \sqrt{\alpha}$ for some $\alpha$, then the probability that at least one measurement yields outcome 1 is at least $\alpha$. However, this is false, as shown by the ``quantum anti-Zeno'' sequence of measurements discussed in the Introduction.  

Here we can use our testing procedures to deal with this issue and give a corrected proof of Theorem \ref{thm:qma1}.  

The corrected Quantum OR bound from~\cite{aaronson06} could also be useful for other applications in quantum information. For example, it can be used to give a proof that the security of certain proposed quantum money schemes~\cite{bennett83} must rest on computational, rather than information-theoretic, assumptions~\cite{aaronson16} (Scott Aaronson, personal communication).

{\bf Remark:} 
It is only the proof in \cite{aaronson06} that is incorrect, but the protocol there may well work.  This is because the measurements are performed in a random order and we do not know of a set of measurements which has the anti-Zeno property for most choices of ordering.

\subsection{Disturbance and $G$-isomorphism}

To provide some intuition for our results, we now outline how applying sequential measurements to quantum states connects to testing $G$-isomorphism. The $G$-isomorphism tester is based on a simple idea: testing whether $g = f\circ\sigma$ for each permutation $\sigma \in G$. (A drawback of this technique is that the testers produced are not time-efficient, but just query-efficient.) We construct $k$ copies of a state $\ket{\psi}$ corresponding to evaluating $f$ and $g$ on every possible input in superposition, for some $k$ to be determined. Given one copy of $\ket{\psi}$, for any $\sigma$ we can distinguish between the cases that $g = f\circ\sigma$, and $g$ is far from $f\circ\sigma$, with success probability lower bounded by a constant. Given $k$ copies of $\ket{\psi}$, we can distinguish between these two cases with exponentially small probability of failure in $k$. We would like to reuse the state $\ket{\psi}^{\otimes k}$ for each test, to avoid making any further queries.

In the ``far from isomorphic'' case, as the probability that the measurement incorrectly says ``isomorphic'' is exponentially small in $k$, $\ket{\psi}^{\otimes k}$ will only be disturbed by an exponentially small amount~\cite{winter99,ogawa02,aaronson06}. However, in the ``isomorphic'' case, it could be the case that, as well as having $g = f\circ\sigma$ for some $\sigma \in G$, we have $g \approx f \circ \tau$ for some $\tau \neq \sigma$. If the test for $\tau$ has a fairly large probability of success, and yet still outputs ``not isomorphic'', the resulting state may be substantially disturbed, implying that the future test for isomorphism under $\sigma$ may incorrectly reject. One way to address this problem is to introduce an additional test for disturbance of the state $\ket{\psi}^{\otimes k}$ (see Appendix \ref{app:dist}): if the state is substantially disturbed at any point during the algorithm, we know that the answer should be ``isomorphic''. Alternatively, we could repeatedly perform a random test and attempt to reset the system to its initial state after each measurement, similarly to Marriott-Watrous gap amplification (see Section \ref{sec:mw}). In either case it turns out to be sufficient to take $k = O((\log |G|)/\epsilon)$ to distinguish between the cases that $f$ and $g$ are $G$-isomorphic, or $\epsilon$-far from $G$-isomorphic.

The strategy of testing each $\sigma$ in turn is similar to the technique used by Ettinger, H\o yer and Knill~\cite{ettinger04} to give a polynomial-query quantum algorithm for the nonabelian hidden subgroup problem (HSP) by testing each subgroup in turn. However, the algorithm of~\cite{ettinger04} did not have the issue with disturbing the state that we need to address here. This was a consequence of the hidden subgroup promise in the standard HSP. In the HSP, we are given access to a function $f$ which is promised to be constant on cosets of some subgroup $H$, and distinct on each coset. The second part of this promise implies that, if we test $f$ for being constant on cosets of any subgroup $H'\neq H$, the test is likely to fail, so the state is left almost undisturbed. In the property-testing scenarios we consider here, we do not have this promise.

It is well known that testing graph isomorphism of rigid graphs (without a promise that ``no'' instances are far from isomorphic) reduces to the HSP for the symmetric group. Our results on $G$-isomorphism, however, do not use a reduction to the property-testing variant of the HSP for $G$. The standard reduction would produce a function $F(\sigma) = f\circ\sigma$, where $\sigma \in G$ and $f\circ\sigma$ represents the table of all values $f(\sigma(x))$; so evaluating $F$ for any given $\sigma$ requires $|X|$ queries. In the case of boolean function isomorphism, for example, evaluating $F$ would require $2^n$ queries, and would destroy any exponential speedup.


\subsection{Organisation}

We present two procedures for determining whether one of a sequence of $n$ measurements accepts a state with high probability: one based on Marriott-Watrous gap amplification, and one based on testing disturbance. The procedures both have similar parameters, and either of them would suffice to prove our main results. The procedure based on testing disturbance\footnote{The first version of this paper only included this procedure.} has somewhat worse constants and a less elegant proof of correctness, so we relegate a description of this to Appendix \ref{app:dist} and focus on the modified Marriott-Watrous procedure, which we now describe.


\section{Modified Marriott-Watrous gap amplification}
\label{sec:mw}


Our procedure will be based on a subroutine which, roughly speaking, performs eigenvalue estimation on a POVM element $\Lambda$ applied to one copy of $\rho$. More precisely, we can use this subroutine to determine if $\rho$ has high support on the space of eigenvectors with large eigenvalues.

Assume that we have one copy of some quantum state $\rho$ and a 2-outcome POVM $\{\Lambda, I-\Lambda\}$, where we are given an explicit decomposition $\Lambda \otimes \ket{0}\bra{0}^{\otimes m} = \Delta \Pi \Delta$, where $m$ is some integer, $\Delta = I \otimes \ket{0}\bra{0}^{\otimes m}$, and $\Pi$ is an orthogonal projector. This is the case in most applications: for example if we are given an explicit circuit description of a measurement corresponding to $\Lambda$ that consists of appending $m$ ancilla qubits in the state $\ket{0}$, applying a unitary $U$, and then making a projective measurement $P$, then the success probability on input $\psi$ is
\[ \bra{\psi} \bra{0}^{\otimes m} U^\dagger P U \ket{\psi}\ket{0}^{\otimes m} = \bra{\psi} \bra{0}^{\otimes m} \Delta \Pi \Delta \ket{\psi}\ket{0}^{\otimes m} \]
where we take $\Pi = U^\dagger P U$. In this case the projector $\Delta$ serves to check that the ancilla qubits are initialized to $\ket{0}$ properly; if they are, the success probability should be equal to $\bra{\psi} \Lambda \ket{\psi}$, and therefore $\Lambda \otimes \ket{0}\bra{0}^{\otimes m} = \Delta \Pi \Delta$ as desired.

Alternatively, if the Naimark extension of $\Lambda$ is given then we also have such a decomposition for $\Lambda$ with $m=1$; or if $\Lambda$ is a projector then we can take a trivial decomposition with $m=0$.

One special case that will be important for us is when we have a sequence of projectors $\Lambda_i$, where $\Lambda_i$ corresponds to the two-outcome measurement $M_i = \{ \Lambda_i, I - \Lambda_i \}$, and we would like to implement the POVM element $\Lambda := \frac{1}{n} \sum_{j=1}^n \Lambda_j$. Define the projector $\Pi = \sum_{i=0}^{n-1} \Lambda_{i+1} \otimes (Q \proj{i} Q^{-1})$, where $Q$ is the quantum Fourier transform on $\Z_n$, and also define $\Delta = I \otimes \proj{0}$. Given quantum circuits which implement each measurement $M_i$, it is easy to write down a circuit that implements the measurement $\{\Pi,I-\Pi\}$. Then
\[ \Delta \Pi \Delta = \left( \frac{1}{n} \sum_{j=1}^n \Lambda_j \right) \otimes \proj{0} = \Lambda \otimes \proj{0}, \]
so this gives us an implementation of $\Lambda$ as desired.

Our procedure is described as Algorithm \ref{alg:mw} below. The algorithm is based on the Marriott-Watrous procedure for in-place amplification for QMA \cite{marriott05}, but we give a simplified procedure that is similar to the ``OR-type repetition procedure'' in \cite{fefferman16}, but with different analysis.

\boxalgm{alg:mw}{Modified Marriott-Watrous gap-amplification procedure}{
\begin{enumerate}
\item Create the state $\rho \otimes \proj{0}^{\otimes m}$. 
\item Repeat $N$ times or until the algorithm accepts:
\begin{enumerate}
\item Perform the projective measurement $\{\Pi,I-\Pi\}$. If the first result is returned, accept.
\item Perform the projective measurement $\{\Delta,I-\Delta\}$. If the second result is returned, accept.
\end{enumerate}
\item Reject.
\end{enumerate}
}

\begin{thm}
\label{thm:mw}
Let $p_{\operatorname{acc}}(N)$ be the acceptance probability of Algorithm \ref{alg:mw} when applied to the measurement operator $\Lambda$ and state $\rho$, and write $P_{\ge \delta}$ for the projector onto $\operatorname{span} \{ \ket{\phi}: \Lambda \ket{\phi} = \lambda \ket{\phi}, \lambda \ge \delta\}$. Then
\[ (1-e^{-1}) \tr P_{\ge \frac{1}{2N}} \rho\;\le\;p_{\operatorname{acc}}(N)\;\le\;2N \tr \Lambda \rho. \]
\end{thm}

\begin{proof}
First assume $\rho$ is pure, $\rho = \proj{\psi}$. Decompose $\ket{\psi}$ into an eigenbasis of $\Lambda$:
\[ \ket{\psi} = \sum_i \alpha_i\ket{\psi_i}, \]
where $\Lambda \ket{\psi_i} = \lambda_i \ket{\psi_i}$ and the $\psi_i$'s are normalized states. Appending the ancilla qubits, we have
\[ \ket{\psi} \otimes \ket{0}^{\otimes m} = \sum_i \alpha_i\ket{\widetilde{\psi}_i} \]
where the states $\ket{\widetilde{\psi}_i}:= \ket{\psi_i} \otimes \ket{0}^{\otimes m}$ are eigenvectors of $\Delta \Pi \Delta$, and moreover $\Delta \ket{\widetilde{\psi}_i} = \ket{\widetilde{\psi}_i}$. Note that
\[ \Delta (I-\Pi) \ket{\widetilde{\psi}_i} = \Delta \ket{\widetilde{\psi}_i} - \Delta \Pi \Delta \ket{\widetilde{\psi}_i} = (1 - \lambda_i) \ket{\widetilde{\psi}_i}. \]
Therefore if Algorithm \ref{alg:mw} does not accept in Step 2, the residual unnormalized state is
\[ (\Delta (I - \Pi))^N \ket{\psi} \otimes \ket{0}^{\otimes m} = \sum_i \alpha_i (1-\lambda_i)^N \ket{\widetilde{\psi}_i} \]
and the probability that Algorithm \ref{alg:mw} accepts (doesn't reach Step 3) is
\[ p_{\operatorname{acc}}(N) = 1 - \sum_i |\alpha_i|^2 (1-\lambda_i)^{2N} = \sum_i |\alpha_i|^2 \left[1- (1-\lambda_i)^{2N}\right]. \]
It is now easy to derive the lower and upper bounds claimed in the theorem. First,
\[ p_{\operatorname{acc}}(N) \ge \sum_{i: \lambda_i \ge 1/(2N)} |\alpha_i|^2 \left[1- (1-\lambda_i)^{2N}\right] > (1-e^{-1}) \tr P_{\ge \frac{1}{2N}} \psi, \]
where we used $(1-a)^{1/a} < e^{-1}$ for any $a > 0$. Second,
\[ p_{\operatorname{acc}}(N) \le \sum_i |\alpha_i|^2 (2N\lambda_i) = 2N \tr \Lambda \psi, \]
where we used $1- (1-\lambda_i)^{2N} \le 2N \lambda_i$. Finally, if $\rho$ is mixed, the claim follows from convexity. This completes the proof.
\end{proof}

We now specialise Theorem \ref{thm:mw} to cases of interest for applications.


\section{Property testing}

The first setting in which we would like to use Theorem \ref{thm:mw} is where we have a sequence of measurements, and would like to determine whether one of them accepts with high probability. We will need the following lemma:

\begin{lem}[Gentle measurement / ``almost as good as new'' lemma~\cite{winter99,ogawa02,aaronson05a,aaronson06}]
\label{lem:gm}
Let $\rho$ be a quantum state and let $0 \le \Lambda \le I$ be a measurement operator. Then
\[ \left\| \rho - \frac{\sqrt{\Lambda} \rho \sqrt{\Lambda}}{\tr \Lambda\rho} \right\|_{\operatorname{tr}} \le \sqrt{\tr (I-\Lambda)\rho}. \]
\end{lem}

We state this lemma somewhat differently to some previous works; a proof of this version can be found in~\cite[Lemma 9.4.1]{wilde13}. For completeness, we also provide a concise proof here.

\begin{proof}
Let $F(\rho,\sigma) := \tr \sqrt{\rho^{1/2} \sigma \rho^{1/2}}$ be the fidelity between quantum states $\rho$ and $\sigma$. We have
\[ F\left(\rho, \frac{\sqrt{\Lambda} \rho \sqrt{\Lambda}}{\tr \Lambda\rho} \right) = \frac{\tr \sqrt{\sqrt{\rho} \sqrt{\Lambda} \rho \sqrt{\Lambda}\sqrt{\rho}}}{\sqrt{\tr \Lambda \rho}} = \frac{\tr \sqrt{\rho} \sqrt{\Lambda} \sqrt{\rho}}{\sqrt{\tr \Lambda \rho}} = \frac{\tr \sqrt{\Lambda} \rho}{\sqrt{\tr \Lambda \rho}} \ge  \frac{\tr \Lambda \rho}{\sqrt{\tr \Lambda \rho}} = \sqrt{\tr \Lambda \rho}, \]
where the second equality follows because $\sqrt{\rho}\sqrt{\Lambda}\sqrt{\rho}$ is positive semidefinite, the third is cyclicity of the trace, and the inequality is $\sqrt{\Lambda} \ge \Lambda$ for $0 \le \Lambda \le I$. Using the inequality $\|\rho-\sigma\|_{\operatorname{tr}} \le \sqrt{1 - F(\rho,\sigma)^2}$, we obtain the claimed result.
\end{proof}

\begin{cor}
\label{cor:seqmeasproptest}
Let $\Lambda_1,\dots,\Lambda_n$ be a sequence of projectors, where $\Lambda_i$ corresponds to the two-outcome measurement $M_i = \{ \Lambda_i, I - \Lambda_i \}$, and fix parameters $\epsilon \le 1/2$, $\delta$. Let $\rho$ be a state such that either there exists $i$ such that $\tr \Lambda_i \rho \ge 1-\epsilon$ (case 1), or $\E_j[\tr \Lambda_j \rho] \le \delta$ (case 2). Then there is a test that uses one copy of $\rho$ and: in case 1, accepts with probability at least $(1-\epsilon)^2/7$; in case 2, accepts with probability at most $4 \delta n$.
\end{cor}

\begin{proof}
We apply Algorithm \ref{alg:mw} to $\Lambda = \frac{1}{n} \sum_j \Lambda_j$ and $\rho$, taking $N = \lceil n/(1-\epsilon) \rceil$. We first consider case 1. Let $Q$ denote the projector onto $\operatorname{span} \{ \ket{\phi}: \Lambda \ket{\phi} = \lambda \ket{\phi}, \lambda \ge 1/(2N)\}$, and write $Q^\perp = I - Q$. To apply Theorem \ref{thm:mw}, we need to lower-bound $\tr Q \rho$. By Lemma \ref{lem:gm},
\[  \tr Q \rho \ge \left\| \rho - \frac{ Q^\perp \rho Q^\perp }{\tr  Q^\perp \rho} \right\|_{\operatorname{tr}}^2. \]
Write $\sigma = Q^\perp \rho Q^\perp / \tr Q^\perp \rho$. Then, as we are in case 1, $\tr \Lambda_i \rho \ge 1-\epsilon$ and hence
\[ \frac{1-\epsilon}{n} \le \frac{\tr \Lambda_i \rho}{n} = \frac{\tr \Lambda_i (\rho-\sigma) + \tr \Lambda_i \sigma}{n} \le \frac{\|\rho-\sigma\|_{\operatorname{tr}}}{n} +  \frac{\tr \Lambda_i \sigma}{n} \le \frac{\|\rho-\sigma\|_{\operatorname{tr}}}{n} + \tr \Lambda \sigma. \]
We have $\tr \Lambda \sigma < 1/(2N) \le (1-\epsilon)/(2n)$ because $\sigma$ is only supported on eigenvectors of $\Lambda$ with eigenvalues less than $1/(2N)$. Rearranging, we obtain that
\[ \|\rho -\sigma\|_{\operatorname{tr}} \ge \frac{1-\epsilon}{2}, \]
so $\tr Q\rho \ge (1-\epsilon)^2/4$ and hence Algorithm \ref{alg:mw} accepts with probability at least $(1-e^{-1})(1-\epsilon)^2/4 \ge (1-\epsilon)^2/7$. In case 2, $\tr \Lambda \rho \le \delta$ by assumption. So by Theorem \ref{thm:mw}, Algorithm \ref{alg:mw} accepts with probability at most $2\delta \lceil n/(1-\epsilon) \rceil \le 4 \delta n$, using $\epsilon \le 1/2$.
\end{proof}


\subsection{Application to testing $G$-isomorphism}

We now apply Corollary \ref{cor:seqmeasproptest} to testing $G$-isomorphism. In fact, we reduce this problem to the following more general question: testing whether a state is an eigenvector of some unitary operator picked from a known set of unitary operators.

\begin{lem}
\label{lem:eigentest}
Assume that we have access to a sequence of controlled unitaries $U_1,\dots,U_n$, and their inverses, and the ability to produce copies of a state $\ket{\psi}$. We are promised that either there exists $i$ such that $U_i \ket{\psi} = \ket{\psi}$, or for all $i$, $|\bracket{\psi}{U_i}{\psi}| \le 1-\epsilon$. There is an algorithm which distinguishes between these two cases using $O((\log n) / \epsilon)$ copies of $\ket{\psi}$ and succeeds with probability at least $2/3$.
\end{lem}

\begin{proof}
Write $\ket{\phi} := (\frac{1}{\sqrt{2}}(\ket{0}+\ket{1})\ket{\psi})^{\otimes k}\ket{0}$ for some $k$ and, for each $i$, define the two-outcome projective measurement $M_i$ by the following sequence of steps:
\begin{enumerate}
\item Apply controlled-$U_i$ on each of the first $k$ registers (controlled on the first qubit within each).
\item Apply a Hadamard gate to each first qubit in each of the first $k$ registers.
\item Apply a controlled-X gate to the last qubit, controlled on all of these $k$ qubits being in the 0 state.
\item Measure the last qubit. Associate outcome 1 with acceptance, 0 with rejection.
\item Invert steps 1-3.
\end{enumerate}
The state after the second step is $(\frac{1}{2}(\ket{0}(I+U_i)\ket{\psi}+\ket{1}(I - U_i)\ket{\psi}))^{\otimes k}\ket{0}$, so the probability of acceptance is
\[ \left(\frac{\| (I+U_i)\ket{\psi} \|^2}{4} \right)^k = \left(\frac{1}{2} + \frac{1}{2} \text{Re} \bracket{\psi}{U_i}{\psi} \right)^k. \]
If $U_i \ket{\psi} = \ket{\psi}$, the measurement accepts with certainty. If $|\bracket{\psi}{U_i}{\psi}| \le 1-\epsilon$, the measurement accepts with probability at most $(1 - \epsilon / 2)^k \le e^{-\epsilon k / 2}$. It is sufficient to take $k = O((\log n) / \epsilon)$ to obtain an acceptance probability in this case which is at most $c/n$ for an arbitrary constant $c > 0$. We can now apply Corollary \ref{cor:seqmeasproptest}. In the former case, we have that the test accepts with probability at least $1/7$. In the latter case, we can choose $c$ such that, by Corollary \ref{cor:seqmeasproptest}, the test accepts with probability at most $1/8$. A constant number of repetitions suffices to distinguish between these two cases with probability at least $2/3$.
\end{proof}

Lemma \ref{lem:eigentest} can be applied to testing $G$-isomorphism. Recall that in this problem we have a group $G$ acting on a set $X$, and would like to distinguish between these two cases: a) there exists $\sigma \in G$ such that $g(x) = f(\sigma(x))$ for all $x \in X$; b) for all $\sigma \in G$, $|\{x \in X: g(x) \neq f(\sigma(x))\}| \ge \epsilon |X|$. Formally, the $G$-isomorphism problem is actually a special case of a quantum problem discussed below (unitary $S$-isomorphism); however, we state and prove it separately for clarity.

\begin{repthm}{thm:giso}
For any $G$, there is a quantum $\epsilon$-tester for $G$-isomorphism which makes $O((\log |G|) / \epsilon)$ queries to $f$ and $g$.
\end{repthm}

\begin{proof}
Write $d(f,g) := |\{x \in X: f(x) \neq g(x)\}|/|X|$. For any function $f:X \rightarrow Y$, define the corresponding state
\[ \ket{f} := \frac{1}{\sqrt{|X|}} \sum_{x \in X} \ket{x} \ket{f(x)}; \]
then $\ip{f}{g} = 1 - d(f,g)$. Also define the unitary operator $U_\sigma$ by $U_\sigma \ket{x} = \ket{\sigma(x)}$ for $\sigma \in G$. Then $U_\sigma\ket{f} = \ket{f\circ\sigma}$. Consider the state $\ket{\psi} = \frac{1}{\sqrt{2}} ( \ket{0} \ket{f} + \ket{1}\ket{g})$ and the operator $U_\sigma'$ which maps $\ket{\psi}$ to $\frac{1}{\sqrt{2}} ( \ket{0} \ket{g \circ \sigma^{-1}} + \ket{1}\ket{f\circ\sigma})$ for any $f$ and $g$. $U_\sigma'$ can easily be implemented using Pauli-X, controlled-$U_\sigma$ and controlled-$U_\sigma^{-1}$ operations. Then
\[ \bracket{\psi}{U_\sigma'}{\psi} = \ip{f\circ\sigma}{g} = 1-d(f\circ\sigma,g). \]
Applying Lemma \ref{lem:eigentest} to the state $\ket{\psi}$ and the unitary operator $U_\sigma'$, we can distinguish between the case that there exists $\sigma$ such that $g = f\circ\sigma$, and the case that $d(g, f\circ\sigma) \ge \epsilon$ for all $\sigma \in G$, with $O((\log |G|) / \epsilon)$ copies of $\ket{\psi}$. Each copy can be created with one query to $f$ and $g$. This proves Theorem \ref{thm:giso}.
\end{proof}


\subsection{Testing quantum states and operations}

We next apply our results to testing properties of quantum states, and then properties of quantum operations. These are all quite straightforward corollaries of previous results in the paper.

\begin{repthm}{thm:teststates}
Let $\mathcal{P}$ be a finite subset of the unit sphere in $\C^d$. Then, for any $\epsilon > 0$, there is a quantum $\eps$-tester for $\mathcal{P}$ which uses $O((\log |\mathcal{P}|) /\epsilon^2)$ copies of the input state.
\end{repthm}

\begin{proof}
Let $k$ be an integer parameter to be determined. We apply Corollary \ref{cor:seqmeasproptest} to the state $\ket{\psi}^{\otimes k}$ and the measurements $\Pi_{\ket{\phi}} = \proj{\phi}^{\otimes k}$, $\ket{\phi} \in \mathcal{P}$. If $\ket{\psi} \in \mathcal{P}$, there exists $\ket{\phi}$ such that $\bra{\psi}^{\otimes k} \Pi_{\ket{\phi}} \ket{\psi}^{\otimes k} = 1$. If $\min_{\ket{\phi} \in \mathcal{P}} \| \psi - \phi \|_{\operatorname{tr}} \ge \epsilon$, then for all $\ket{\phi} \in \mathcal{P}$,
\[ \bra{\psi}^{\otimes k} \Pi_{\ket{\phi}} \ket{\psi}^{\otimes k} = |\ip{\psi}{\phi}|^{2k} = (1- \| \psi - \phi \|_{\operatorname{tr}} ^2)^k \le (1-\epsilon^2)^k. \]
In the former case, the algorithm of Corollary \ref{cor:seqmeasproptest} accepts with probability at least $1/7$. In the latter case, it accepts with probability at most $4|\mathcal{P}|(1-\epsilon^2)^{k}$, so it is sufficient to take $k = O((\log |\mathcal{P}|) / \epsilon^2)$ to make the acceptance probability at most $1/8$. Taking the median of $O(1)$ runs is enough to distinguish these two cases except with probability at most $1/3$.
\end{proof}

We now turn to testing properties of quantum operations. Lifting our results on testing quantum states to testing unitary operators on $\C^d$ is based on the following connection, known as the Choi-Jamio\l kowski isomorphism. Let $\ket{\Phi} = \frac{1}{\sqrt{d}} \sum_{i=1}^d \ket{i} \ket{i}$, and for any $U \in U(d)$, define
\[ \ket{U} = (U \otimes I) \ket{\Phi} = \frac{1}{\sqrt{d}} \sum_{i,j=1}^d U_{ij} \ket{i} \ket{j}. \]
Then it is easy to see that
\be \label{eq:cjcons} \hsip{U}{V} = \ip{U}{V},\;\;\;\; (A \otimes B) \ket{V} = \ket{AVB^T}. \ee
As a consequence of the first equality, $D(U,V) := \sqrt{1 - |\hsip{U}{V}|^2} = \| \proj{U} - \proj{V} \|_{\operatorname{tr}}$. A copy of $\ket{U}$ can be created with one use of $U$.

\begin{repcor}{cor:testunitaries}
Let $\mathcal{P}$ be a finite subset of $U(d)$. Then, for any $\epsilon > 0$, there is a quantum $\eps$-tester for $\mathcal{P}$ which makes $O((\log |\mathcal{P}|) /\epsilon^2)$ uses of the input unitary operator.
\end{repcor}

\begin{proof}
Apply Theorem \ref{thm:teststates} to test membership of $\ket{U}$ in the set $\mathcal{P}' = \{ \ket{V}: V \in \mathcal{P}\}$. The test uses $O((\log |\mathcal{P}|) /\epsilon^2)$ copies of $\ket{U}$, each of which can be constructed with one use of $U$.
\end{proof}

Next we consider the property of unitary $S$-isomorphism. Recall that in this problem we have a set $S = U_1,\dots,U_n$ of unitary operators, and two unitary operators $V$ and $W$. We would like to distinguish between two cases: either there exists $U \in S$ such that $U V U^\dag = W$, or for all $U \in S$, $D(U V U^\dag, W) \ge \epsilon$.

\begin{repthm}{thm:uiso}
For any $S$ and any $\epsilon > 0$, there is a quantum $\epsilon$-tester for unitary $S$-isomorphism which makes $O((\log |S|) / \epsilon^2)$ uses of the input unitaries $V$ and $W$.
\end{repthm}

\begin{proof}
The argument is similar to the proof of Theorem \ref{thm:giso}. We can produce a copy of the state $\ket{\psi} = \ket{V}\ket{W}$ with a single use of each of $V$ and $W$. Similarly, by (\ref{eq:cjcons}), for any $U$ we can implement an operation $U'$ mapping $\ket{\psi}$ to $\ket{\psi'} = \ket{U^\dag W U}\ket{UVU^\dag}$ by applying $U \otimes U^*$ to the first register, and $U^\dag \otimes U^T$ to the second register; and then swapping the two registers. If $U V U^\dag = W$, then $\ket{\psi'} = \ket{\psi}$. On the other hand, if there exists $U \in S$ such that $D(U V U^\dag,W)=\epsilon$, then
\[ |\bracket{\psi}{U'}{\psi}| =  |\hsip{UVU^\dag}{W}|^2 = 1-\epsilon^2. \]
By Lemma \ref{lem:eigentest}, these two cases can be distinguished with $O((\log |S|) / \epsilon^2)$ uses of $V$ and $W$.
\end{proof}

The apparently worse scaling with $\epsilon$ of this result compared with Theorem \ref{thm:giso} is an artifact of the distance measure used being defined differently.


\subsection{Testing genuinely multipartite entanglement}

Our final quantum tester is for the property of not possessing genuine multipartite entanglement -- i.e.\ a quantum state being product across some partition of the qubits into two parts.

\begin{repthm}{thm:genuine}
There is a quantum $\eps$-tester for the property of an $n$-partite state being product across some cut.  The tester uses $O(n/\eps^2)$ copies of the state.
\end{repthm}

\begin{proof}
Suppose we are given $\ket{\psi}^{\otimes k}$ for some $n$-partite state $\ket\psi$.  For any fixed proper nonempty $S\subset [n]$ let $D(S)$ denote the minimum distance of $\ket\psi$ to a state of the form $\ket{\alpha}_S\ot \ket{\beta}_{S^c}$.  According to Lemma 20 of \cite{harrow13} (see also \cite{wei03}), there is a test using two copies of $\ket\psi$ that will accept with probability $1 - \Theta(D(S)^2)$. In particular if $\ket\psi$ is product across $S:S^c$ then it will accept with certainty, and if it is $\eps$-far from product across $S:S^c$ then it will accept with probability $1-\Omega(\eps^2)$. Taking $k=O(n/\eps^2)$ we can reduce this acceptance probability to $2^{-\Omega(n)}$. Since there are $2^{n-1}-1$ ways to partition $[n]$ into two pieces, our result follows from Corollary \ref{cor:seqmeasproptest}. 
\end{proof}


\section{De-Merlinizing quantum protocols}

We finally apply Algorithm \ref{alg:mw} to prove Theorem \ref{thm:qma1}. The key technical ingredient can be stated as follows. Let $\Gamma$ be a measurement operator corresponding to the 2-outcome POVM measurement $\{ \Gamma, I - \Gamma \}$ on a bipartite Hilbert space $\mathcal{H}_A \otimes \mathcal{H}_B$. Let $\{\ket{1},\dots,\ket{d}\}$ be an orthonormal basis for $\mathcal{H}_B$. For all $j \in \{1,\dots,d\}$, let $\Gamma_j$ be the measurement operator on $\mathcal{H}_A$ corresponding to the 2-outcome POVM $\{\Gamma_j, I - \Gamma_j\}$ induced by applying $\Gamma$ to $\mathcal{H}_A \otimes \ket{j}$.

Fix a pure state $\ket{\psi} \in \mathcal{H}_A$. We would like to distinguish between the following two cases:
\begin{enumerate}
\item There exists $\sigma$ in $\mathcal{H}_B$ such that $\Gamma$ accepts with probability at least $\eta > 0$ when applied to $\psi \otimes \sigma$.
\item For all states $\sigma$ in $\mathcal{H}_B$, $\Gamma$ accepts with probability at most $\zeta$ when applied to $\psi \otimes \sigma$.
\end{enumerate}
To do so, we apply Algorithm \ref{alg:mw} to $\ket{\psi}$ and the operator $\Lambda = \frac{1}{d} \sum_{j=1}^d \Lambda_j$, taking $N = \lceil d / \eta \rceil$.

\begin{cor}
\label{cor:demerlin}
In case 1, Algorithm \ref{alg:mw} accepts with probability at least $\eta^2/7$. In case 2, Algorithm \ref{alg:sequential} accepts with probability at most $2 \zeta \lceil d/\eta \rceil$.
\end{cor}

\begin{proof}
The proof is similar to that of Corollary \ref{cor:seqmeasproptest}. We first observe that, for an arbitrary state $\phi$,
\beas \tr \Lambda \phi &=& \E_j \left[ \tr \Lambda (\phi \otimes \proj{j}) \right] =  \tr \Lambda \left(\phi \otimes \frac{I}{d} \right)  \ge \frac{\tr \Lambda (\phi \otimes \sigma)}{d}\\
&\ge& \frac{\tr \Lambda (\psi \otimes \sigma) -\| \phi - \psi \|_{\operatorname{tr}} }{d} \ge \frac{\eta -\| \phi - \psi \|_{\operatorname{tr}} }{d},
\eeas
implying that for any state $\ket{\phi}$ such that $\tr \Lambda \phi \le \eta / (2d)$, $\|\phi - \psi\|_{\operatorname{tr}} \ge \eta/2$. 
Let $Q$ denote the projector onto $\operatorname{span} \{ \ket{\phi}: \Lambda \ket{\phi} = \lambda \ket{\phi}, \lambda \ge 1/(2N)\}$, set $Q^\perp = I - Q$, and take $\ket{\phi} = Q^\perp \ket{\psi} / \|Q^\perp \ket{\psi}\|$. As $N = \lceil d/\eta \rceil$, $\tr Q^\perp \phi \le \eta/(2d)$. By the same ``gentle measurement'' argument as used in the proof of Corollary \ref{cor:seqmeasproptest}, the fact that $\|\phi - \psi\|_{\operatorname{tr}} \ge \eta/2$ implies that $\tr Q \psi \ge \eta^2/4$. So by Theorem \ref{thm:mw}, Algorithm \ref{alg:mw} accepts with probability at least $(1-e^{-1})\eta^2/4 \ge \eta^2/7$. In case 2, $\tr \Lambda \psi = \tr \Lambda(\psi \otimes I/d) \le \zeta$, so it is immediate from Theorem \ref{thm:mw} that Algorithm \ref{alg:mw} accepts with probability at most $2 \zeta \lceil d / \eta \rceil$.
\end{proof}

It is now straightforward to give a corrected proof of Theorem \ref{thm:qma1} by using Corollary \ref{cor:demerlin} within the framework of Aaronson~\cite{aaronson06}. We will use the following lemma from~\cite{aaronson06}:

\begin{lem}[Aaronson~\cite{aaronson06}]
\label{lem:aaramp}
Suppose Bob receives an $a$-qubit message $\ket{\psi}$ from Alice and a $w$-qubit message $\ket{\phi}$ from Merlin, where $w \ge 2$. Let $A = O(aw\log^2 w)$ and $W = O(w \log w)$. Then by using $A$ qubits from Alice and $W$ qubits from Merlin, Bob can amplify his soundness error to $5^{-W}$ while keeping his completeness error $1/3$.
\end{lem}

\begin{repthm}{thm:qma1}
For all partial or total boolean functions $f$, and all $w \ge 2$,
\[ Q^1(f) = O(\operatorname{QMA}^1_w(f) \cdot w \log^2 w). \]
\end{repthm}

\begin{proof}
Let $\Gamma$ be the measurement corresponding to Bob's amplified algorithm from Lemma \ref{lem:aaramp} and let $\ket{\psi_x}$ be the state of Alice's register. We know that, if $f(x,y) = 1$, there exists a state $\ket{\phi}$ of the witness register such that $\Gamma$ accepts $\psi_x \otimes \phi$ with probability at least $2/3$. On the other hand, if $f(x,y)=0$, then for all witness states $\ket{\phi}$, $\Lambda$ accepts $\psi_x \otimes \phi$ with probability at most $5^{-W}$. Inserting these parameters within Corollary \ref{cor:demerlin} and using $d=2^W$, we find that in the former case Algorithm \ref{alg:mw} accepts with probability at least $4/63$, and in the latter case accepts with probability at most $4 \cdot 5^{-W} \cdot 2^W = o(1)$. The two cases can therefore be distinguished with $O(1)$ repetitions.
\end{proof}

\subsection*{Acknowledgements}

We would like to thank Noah Linden for suggesting the application to testing whether one unitary operator is a permutation of another, Mark Wilde for pointing out references \cite{gao15,sen12,wilde13a} and Scott Aaronson for helpful comments on a previous version.
AM was supported by EPSRC Early Career Fellowship EP/L021005/1.
AWH was funded by NSF grants CCF-1629809 and CCF-1452616. CYL is supported by the Department of Defense.


\appendix

\section{An alternative protocol via testing disturbance}
\label{app:dist}

In this appendix we describe a alternative approach towards determining whether one of a sequence of $n$ measurements accepts an input state, based around testing disturbance of the input state. We will need the following result regarding sequences of measurements:

\begin{lem}[Improved quantum union bound~\cite{gao15}]
\label{lem:quantumunion}
Let $\rho$ be a mixed state, and let $M_1,\dots,M_T$ be a sequence of 2-outcome projective measurements. Suppose each $M_t$ yields outcome 1 with probability at most $\epsilon$ when applied to $\rho$. Then if we apply $M_1,\dots,M_T$ in sequence to $\rho$, the probability that at least one measurement yields outcome 1 is at most $4T\epsilon$.
\end{lem}

Lemma \ref{lem:quantumunion}, which is due to Gao~\cite{gao15}, improves previous bounds of a similar nature~\cite{winter99,ogawa02,aaronson06,sen12,wilde13a} up to quadratically.

We assume that we have one copy of some state $\rho$, and have access to quantum circuits which allow us to coherently implement each of a sequence of 2-outcome POVMs specified by projectors $\Lambda_1,\dots,\Lambda_n$, where each $\Lambda_i$ corresponds to the measurement $M_i = \{\Lambda_i,I-\Lambda_i\}$, and the first outcome is associated with acceptance, the second with rejection. We further assume that $\eta$, $\zeta$ are picked such that exactly one of the following two cases holds:
\begin{enumerate}
\item $\rho$ is pure and, for all pure states $\ket{\phi}$ such that $\| \rho - \phi \|_{\operatorname{tr}} \le \eta$, we have $\E_j [ \tr \Lambda_j \phi ] \ge \eta / n$ (``the average probability of acceptance is quite high for all states relatively close to $\rho$'').
\item For all $j$, $\tr \Lambda_j \rho \le \zeta$ (``the probability that any measurement accepts $\rho$ is low'').
\end{enumerate}
Our task is to accept in the first case, and reject in the second. The first case may seem somewhat unintuitive, but we state it in this way so that it encompasses all our applications.

We use Algorithm \ref{alg:sequential} below to complete this task. The intuition behind this algorithm is as follows, in the case that $\rho = \proj{\psi}$ and $\eta = \Theta(1)$. Throughout the algorithm, the state of the system is of the form $\alpha \ket{0}\ket{\psi} + \beta \ket{1}\ket{\widetilde{\psi}}$ for some state $\ket{\widetilde{\psi}}$. Assume we are in case 1 above. If $\ket{\widetilde{\psi}} \approx \ket{\psi}$ and $\beta$ is not too small, the next random choice of measurement will accept with fairly high probability (roughly $\Omega(1/n)$). On the other hand, if $\ket{\widetilde{\psi}}$ is far from $\ket{\psi}$ or $\beta$ is small, the test in step 2a would accept with high probability if it were performed. So the overall probability that the test accepts at this stage is $\Omega(1/n)$ in either case; repeating $O(n)$ times, the overall acceptance probability is $\Omega(1)$. On the other hand, if we are in case 2, we can use Lemma \ref{lem:quantumunion} to infer that after making $O(n)$ measurements the overall acceptance probability is $O(n\zeta)$.

\boxalgm{alg:sequential}{Sequential measurement test}{
\begin{enumerate}
\item Create the state $\proj{+} \otimes \rho$, where $\ket{+} = \frac{1}{\sqrt{2}}( \ket{0} + \ket{1})$.
\item Repeat the following $k := \lceil 5n / \eta+5/\eta^2 \rceil$ times:
\begin{enumerate}
\item With probability $1 / (\eta n + 1)$, perform a Hadamard gate on the first qubit and measure it in the computational basis. If the outcome is 0, reject; otherwise, accept.
\item Pick $j \in [n]$ uniformly at random.
\item Perform the measurement $\{\proj{1} \otimes \Lambda_j,I-\proj{1} \otimes \Lambda_j\}$. If the first outcome is returned,  accept. Otherwise, retain the residual state of the two registers.
\end{enumerate}
\item Reject.
\end{enumerate}
}

We now prove the correctness of Algorithm \ref{alg:sequential} more formally.

\begin{thm}
\label{thm:prot}
In case 1, Algorithm \ref{alg:sequential} accepts with probability at least $\eta^2/7 - O(1/n)$. In case 2, Algorithm \ref{alg:sequential} accepts with probability at most $2 \lceil 5n/\eta+5/\eta^2 \rceil \zeta$.
\end{thm}

\begin{proof}
First consider case 1, in which we would like the algorithm to accept. In this case we assume that $\rho$ is pure, so write $\rho = \proj{\psi}$. The algorithm accepts if and only if either the measurement in step 2a is made and the outcome is 1, or the first measurement outcome in step 2c is obtained. Call either of these a ``good'' measurement outcome.

The overall state of the algorithm at the start of the $i$'th step of the loop can be written as $\alpha_i \ket{0}\ket{\psi} + \beta_i \ket{1}\ket{\psi_i}$ for some normalised state $\ket{\psi_i}$ and some $\alpha_i,\beta_i \in \C$ such that $|\alpha_i|^2 + |\beta_i|^2=1$, with $\ket{\psi_1} = \ket{\psi}$ and $\alpha_1 = \beta_1 = 1/\sqrt{2}$. For any such state, the probability that the measurement in step 2a would return an outcome of 1, if it were made, is
\be \label{eq:ipbound} \frac{1}{2} \| \alpha_i \ket{\psi} - \beta_i \ket{\psi_i} \|^2 = \frac{1}{2}(1 - 2\Re(\alpha_i^* \beta_i \ip{\psi}{\psi_i})). \ee
We say that $\ket{\psi_i}$ is disturbed if $\| \psi_i - \psi \|_{\operatorname{tr}} \ge \eta$, and undisturbed otherwise. First assume that $\ket{\psi_i}$ is undisturbed. Then the probability that the measurement in step 2c is made and the first measurement outcome is obtained is
\be \left(1 - \frac{1}{\eta n + 1} \right) |\beta_i|^2 \E_j[\tr \Lambda_j \psi_i] \ge \left(1 - \frac{1}{\eta n + 1} \right) |\beta_i|^2 \frac{\eta}{n}  
\label{eq:ud-2c}\ee
as we are in case 1. By (\ref{eq:ipbound}), the probability that the measurement in step 2a is made and an outcome of 1 is obtained is lower-bounded by
\be \frac{1}{\eta n + 1} \left(1-2|\beta_i|\sqrt{1-|\beta_i|^2} |\ip{\psi}{\psi_i}| \right) \ge \frac{1}{\eta n + 1} \left(1-2|\beta_i|\sqrt{1-|\beta_i|^2} \right). 
\label{eq:ud-2a}\ee
If $|\beta_i| \ge 1/\sqrt{5}$, then \eq{ud-2c} is $\geq \eta^2 / (5(\eta n+1))$; if $|\beta_i| \le 1/\sqrt{5}$ then \eq{ud-2a} is $\geq 1/(5(\eta n + 1))$. As $\eta^2 \le 1$, the first bound is always lower, so the probability that a good measurement outcome is received if $\ket{\psi_i}$ is undisturbed obeys the overall lower bound of $\eta^2 / (5(\eta n+1))$.

On the other hand, if $\ket{\psi_i}$ is disturbed, by (\ref{eq:ipbound}) the probability that the measurement in step 2a is made and returns an outcome of 1 is lower-bounded by
\[ \frac{1}{\eta n + 1} \cdot \frac{1}{2}(1 - 2|\alpha_i||\beta_i| |\ip{\psi}{\psi_i}|) \ge \frac{1 - |\ip{\psi}{\psi_i}|}{2(\eta n+1)}  \ge \frac{1 - |\ip{\psi}{\psi_i}|^2}{4(\eta n+1)} \ge \frac{\eta^2}{4(\eta n + 1)}. \]
Therefore, at the $i$'th step of the loop, the probability of a good measurement outcome is at least $p := \eta^2 / (5(\eta n+1))$ whether or not $\ket{\psi_i}$ is disturbed. The probability that the protocol fails at any given step -- by incorrectly rejecting -- is at most $q := 1 / (\eta n + 1)$. So the probability that the protocol terminates with a good measurement outcome occurring before failure is lower-bounded by
%
\[ (1-q)p + (1-p)(1-q)^2p + (1-p)^2(1-q)^3 p + \dots + (1-p)^{k-1}(1-q)^k p = p(1-q)\left( \frac{1-(1-p)^k (1-q)^k}{1-(1-p)(1-q)}\right). \]
We have $(1-q)^k \le (1-p)^k \le e^{-pk} \le 1/e$, so the overall probability of success is lower-bounded by
%
\[ \left(1-\frac{1}{e^2}\right) \frac{p-pq}{p+q-pq} = \left(1-\frac{1}{e^2}\right) \frac{\eta^3 n}{5 \eta n + \eta^3 n + 5} \ge \frac{\eta^2}{7} - O\left(\frac{1}{n}\right).\]
%
Now consider case 2. By assumption, $\Lambda_j$ accepts $\rho$ with probability at most $\zeta$ for all $j$, so the measurement operator $\proj{1} \otimes \Lambda_j$ accepts the starting state $\proj{+} \otimes \rho$ with probability at most $\zeta/2$. Also, the measurement in step 2a would reject the starting state with certainty. At most $k$ measurements are made in the algorithm. By the quantum union bound (Lemma \ref{lem:quantumunion}), the probability that any measurement made in the algorithm leads to acceptance is upper-bounded by $2k\zeta$.
\end{proof}

An alternative protocol is to start with two copies of the input state, perform the measurements $M_j$ on only the first copy, and test disturbance between the two copies using the swap test~\cite{buhrman01}. This would avoid the need for controlled measurements, but would require an additional copy of the input state.

Given Theorem \ref{thm:prot}, it is easy to show variants of Corollary \ref{cor:seqmeasproptest} and Corollary \ref{cor:demerlin} with slightly worse constants, which imply the rest of the results in the paper.


\section{Classical lower bound for testing linear isomorphism}
\label{app:linisobound}

\begin{prop}
There is a universal constant $\epsilon > 0$ such that any classical $\epsilon$-tester for linear isomorphism of two unknown boolean functions must make $\Omega(2^{n/2})$ queries.
\end{prop}

\begin{proof}
The proof is very similar to the lower bound on the property-testing variant of Simon's problem~\cite{buhrman08}. By the Yao principle, it is sufficient to bound the success probability of deterministic algorithms which distinguish between the following two distributions:
\begin{itemize}
\item $\mathcal{D}_{\text{yes}}$: $f:\{0,1\}^n \rightarrow \{0,1\}$ is picked uniformly at random, $A \in GL_n(\F_2)$ is picked uniformly at random, and $g$ is defined by $g(x) = f(Ax)$.
\item $\cD_{\text{no}}$: 
$f,g:\{0,1\}^n \rightarrow \{0,1\}$ are each picked uniformly at random, conditioned on $g(0^n) = f(0^n)$ and
\be \forall A\in GL_n(\F_2)\quad |\{x : g(x) = f(Ax)\}| \leq (1-\eps) 2^n.
\label{eq:eps-far}\ee
\end{itemize}
To be precise, the algorithm is given an input picked from the distribution $\mathcal{D} = \frac{1}{2}\left( \mathcal{D}_{\text{yes}} + \mathcal{D}_{\text{no}} \right)$, and is asked to determine whether it was picked from $\mathcal{D}_{\text{yes}}$ or $\mathcal{D}_{\text{no}}$. 
Since $\cD_{\text{no}}$ is not easy to analyze, we first argue that it is close to a much simpler distribution.  Define $\cD_{\text{rand}}$ to be the uniform distribution over $f,g:\{0,1\}^n\ra \{0,1\}$ subject only to the constraint that $g(0^n)=f(0^n)$.   We claim that $\cD_{\text{rand}}$ satisfies \eq{eps-far} with high probability, which will imply that $\cD_{\text{rand}}$ and $\cD_{\text{no}}$ are close in variational distance.   Indeed, fix a choice of $f$ and $A$.  Let $f\circ A$ denote the function $x\mapsto f(Ax)$.  Then the probability that a random $g$ agrees with $f\circ A$ in a $\geq 1-\eps$ fraction of positions is $\approx 2^{-(1-H_2(\eps))2^n}$ where $H_2(\eps) = -\eps\log(\eps) - (1-\eps)\log(1-\eps)$.  Since there are $\leq 2^{n^2}$ choices of $A$, the probability that a random $g$ fails to satisfy \eq{eps-far} is at most
\[ \exp\left(n^2 - 2^n (1-H_2(\eps))\right).\]
For $\eps < 1/2$ and sufficiently large $n$ this probability is negligible.  We now proceed as though the input were chosen from the distribution
$\frac{1}{2}\left( \mathcal{D}_{\text{yes}} + \mathcal{D}_{\text{rand}} \right)$.

Now consider any deterministic decision tree which queries positions in $f$ and $g$, without loss of generality querying a distinct position at each step. The values $f(0^n)$ and $g(0^n)$ give no useful information for the algorithm to distinguish between $\mathcal{D}_{\text{yes}}$ and $\mathcal{D}_{\text{rand}}$, so we can assume that they are never queried. Thus, in a ``no'' instance, the response to queries is always uniformly random. In a ``yes'' instance, the response to a new query, say, $g(x')$ is uniformly random unless there exists $x$ such that $f(x)$ has been previously queried and $x' = Ax$. If every query by the algorithm receives a uniformly random response, the algorithm cannot distinguish this from a ``no'' instance.

For any sequence of $k$ previous queries, the probability (over the random choice of $A$) that the next query corresponds to a pair $x' = Ax$ of this form is at most $k / (2^n - 1)$ by a union bound. Therefore, the probability that an algorithm making $k$ queries has found such a pair at any point in its execution is $O(k^2 / 2^n)$. So, to achieve success probability $2/3$, it is necessary to make $k = \Omega(2^{n/2})$ queries.
\end{proof}


\bibliographystyle{plain}
\bibliography{disturbing}

\end{document}